\documentclass[11pt]{article}

\usepackage{palatino}
\usepackage{fullpage}
\usepackage[pagebackref,letterpaper=true,colorlinks=true,pdfpagemode=none,urlcolor=blue,linkcolor=blue,citecolor=violet,pdfstartview=FitH]{hyperref}
\usepackage{microtype}

\usepackage{paralist}
\usepackage{makecell}
\usepackage{algorithm}
\usepackage{algpseudocode}

\usepackage{enumerate}
\usepackage{xcolor}
\usepackage{verbatim}
\usepackage{array}
\usepackage{mathdots}
\usepackage{amsthm,amssymb,amsmath}
\usepackage{mathtools}
\usepackage{algorithm}
\usepackage{algpseudocode}
\usepackage{mathrsfs}

\newtheorem{theorem}{Theorem}
\theoremstyle{definition}
\newtheorem{definition}{Definition}
\newtheorem{lemma}[theorem]{Lemma}
\newtheorem{claim}[theorem]{Claim}
\newtheorem{fact}[theorem]{Fact}
\newtheorem*{theorem*}{\bf Informal Theorem}
\newcommand{\etal}{{\em et al}.~}

\newcommand{\calF}{\mathscr{F}}
\newcommand{\cA}{\mathcal{A}}
\newcommand{\cB}{\mathcal{B}}
\newcommand{\cF}{\mathscr{F}}
\newcommand{\cI}{\mathscr{I}}
\newcommand{\cJ}{\mathscr{J}}
\newcommand{\cM}{\mathcal{M}}
\newcommand{\cP}{\mathscr{P}}
\newcommand{\cQ}{\mathcal{Q}}

\newcommand{\R}{\mathbb R}
\newcommand{\eps}{\varepsilon}

\newcommand{\poly}{\mathrm{poly}}
\newcommand{\cov}{\mathsf{cov}}
\newcommand{\opt}{{\mathsf{opt}}}
\newcommand{\hopt}{{\widehat{\opt}}}
\newcommand{\sfm}{m}

\newcommand{\sfv}{\mathsf{V}}
\newcommand{\val}{{\mathsf{val}}}
\newcommand{\chld}{{\mathsf{Chld}}}
\newcommand{\rfc}{{Robust $\mathscr{F}$-Supplier}}
\newcommand{\fpcm}{{$\cF$-PCM}}

\def\fpcm{\calF\textrm{-PCM}}
\def\fpcf{\calF\textrm{-PCF}}
\newcommand{\reps}{\mathsf{Reps}_{\cov}}
\newcommand{\CovP}{\mathscr{P}^\cI_\cov}

\begin{document}
\title{Generalized Center Problems with Outliers}
\author{Deeparnab Chakrabarty\footnote{Department of Computer Science, Dartmouth College, Email: deeparnab@dartmouth.edu}\and  Maryam Negahbani\footnote{Department of Computer Science, Dartmouth College Email: maryam@cs.dartmouth.edu} }
\date{}
\maketitle

\begin{abstract}
We study the $\calF$-center problem with outliers: given a metric space $(X,d)$, a general down-closed family $\calF$ of subsets of $X$,
and a parameter $m$, we need to locate a subset $S\in \calF$ of centers such that the maximum distance among the closest $m$ points in $X$ to 
$S$ is minimized. 

Our main result is a {\em dichotomy theorem}. Colloquially, we prove that there is an efficient $3$-approximation for the $\calF$-center problem with outliers
 if and only if we can 
efficiently optimize a {\em poly-bounded} linear function over $\calF$ subject to a partition constraint.
One concrete upshot of our result is a polynomial time $3$-approximation for the knapsack center problem with outliers for which no (true) approximation algorithm was known. 
\end{abstract}

\section{Introduction}
The $k$-center problem is a classic discrete optimization problem with numerous applications.
Given a metric space $(X,d)$ and a positive integer $k$, the objective is to choose a subset $S\subseteq X$ of at most $k$ points
such that $\max_{v\in X} d(v,S)$ is minimized, where $d(v,S) = \min_{u\in S} d(v,u)$. 
Informally, the problem is to open $k$ centers to serve all points, minimizing the maximum distance to service.
This problem has been studied for at least 50 years~\cite{Hakimi64, Hakimi65}, is NP-hard to approximate 
to a factor better than $2$~\cite{HN79}, and has a simple $2$-approximation algorithm~\cite{Gon85, HS85a}.

In many applications one is interested in a nuanced version of the problem where instead of serving all points in $X$, the objective is to serve at least a certain number of points. This is the so-called $k$-center with outliers version, or the 
{\em robust $k$-center} problem. This problem was first studied by Charikar~\etal in~\cite{CKMN01} 
which gives a $3$-approximation for the problem. A best possible $2$-approximation algorithm was recently given by Chakrabarty~\etal in~\cite{CGK16} (see also the paper~\cite{HPST17} by Harris~\etal). 

Another generalization of the $k$-center problem arises when the location of centers has more restrictions.
For instance, if each point in $X$ has a different weight and the constraint is that the total weight of centers opened is at most $k$.
This problem, now called the {\em knapsack center} problem, was studied by Hochbaum and Shmoys in~\cite{HS86} which gives a $3$-approximation for the problem. To take another instance, $X$ could be vectors in high dimension and the centers picked need to be linearly independent vectors. This motivates the {\em matroid center} problem where the set of centers must be an independent set in a matroid. Chen \etal give a $3$-approximation for this problem in~\cite{CLLW13}.

Naturally, the two aforementioned generalizations can be taken together. Indeed, for the {\em robust matroid center} problem, that is, the problem of picking centers which are an independent set and only $m$ points need to be served, there is a $7$-approximation algorithm in~\cite{CLLW13}. This was recently improved to a $3$-approximation in~\cite{HPST17}. The {\em robust knapsack center} problem, however, has had no non-trivial approximation algorithm till this work. Both~\cite{CLLW13} and~\cite{HPST17} give {\em bi-criteria} $3$-approximation algorithms which violate the knapsack constraint by $(1+\eps)$ (the running time of their algorithm is exponential in $1/\eps$). 

\paragraph*{Our Contributions}

Motivated by the state-of-affairs of the robust knapsack center problem, we study a broad generalization of the problems mentioned above.
Let $\calF$ be a general down-closed\footnote{if $A\in \calF$ and $B\subseteq A$, then $B\in \calF$.} family of subsets over $X$.
In the {\em robust $\calF$-center} problem we are given a metric space $(X,d)$, a parameter $m$, and the objective is to select a subset $S\in \calF$ such that $\min_{T\subseteq X, |T|=m} \max_{v\in T} d(v,S)$ is minimized. That is, the maximum distance of service of the closest $m$ points is minimized. 

Observe that if $\calF := \{A: w(A) \leq k\}$ then we get the robust knapsack center problem, and if $\calF$ is the collection of independent sets of a matroid, then we get the robust matroid center problem. But this generalization captures a host of other problems. For instance, one can consider multiple (but constant) knapsack constraints. Indeed, this was studied in both~\cite{HS86} and~\cite{CLLW13}. The former\footnote{The complete proofs can be found in the STOC 1984 version of ~\cite{HS86}} only looks at the version {\em without} outliers and gives a polynomial time $3$-approximation in the case when the weights are all polynomially bounded. The latter proves that when the weights are not polynomially bounded, there can be no approximation algorithm via a reduction to the {\sc Subset Sum} problem, and gives a $3$-approximation violating each knapsack constraint by at most $(1+\eps)$ multiplicative factor.

 Another instance is a single knapsack constraint along with a single matroid constraint. To our knowledge, this problem has not been studied earlier even in the case when outliers are not allowed.
This problem seems natural: for instance, when the points are high dimensional vectors with weights and the collection of centers needs to be a linearly independent set with total weight at most $k$. 

The complexity of the robust $\calF$-center problem naturally depends on the complexity of $\calF$. To understand this, we define the following optimization problem which depends only on the set-system $(X,\calF)$. We call it the $\calF${\em-maximization under partition constraints} or simply $\fpcm$. In this problem, one is given an arbitrary partition $\cP$ of $X$ along with $\calF$, and a {\em poly-bounded} (the range is at most a polynomial in $|X|$) value $\val(x)$ on each $x\in X$. The objective is to find a set $S\in \calF$
maximizing $\val(S)$ such that $S$ contains at most one element from each part of $\cP$. Our main result stated colloquially (and formally stated as Theorem~\ref{thm:1} and Theorem~\ref{thm:2} in Section~\ref{sec:prelims}) is the following dichotomy theorem\footnote{We are deliberately being inaccurate here. We should state the theorem for the more general {\em supplier} version where the set $X$ is partitioned into $F\cup C$ and only the points in $C$ need to be covered and only the centers in $F$ can be opened. Being more general, the algorithmic results are therefore stronger. On the other hand, we weren't able (and didn't try too hard) to make our hardness go through for the center version. In the Introduction we stick with the center version and switch to the supplier in the more formal subsequent sections.}.

\begin{theorem*} 
{\em For any down-closed family $(X,\calF)$, the robust $\calF$-center problem has an efficient $3$-approximation algorithm if the $\fpcm$ problem
can be solved in polynomial time. Otherwise, there is no efficient {\em non-trivial} approximation algorithm for the robust $\calF$-center problem.}
\end{theorem*}
Note that in general, we are not concerned about how $\calF$ is represented, because the only place the algorithm checks if a set $S$ is in $\calF$ is perhaps for solving the $\fpcm$ problem. So one can choose a representation that works best for the $\fpcm$ solver.

A series of corollaries follow from the above theorem. These are summarized in Table~\ref{tbl:1}.
\begin{asparaitem}
	\item When $\calF =\{A: w(A) \leq k\}$, the $\fpcm$ problem can be solved in polynomial time via dynamic programming. This crucially uses that the $\val$
	is poly-bounded. Therefore we get a $3$-approximation for the robust knapsack center problem. (Theorem~\ref{thm:rkn})
	\item When $\calF$ is the independent set of a matroid, then the $\fpcm$ problem is a matroid intersection problem. Therefore  we get a $3$-approximation for the robust matroid center problem recovering the result from~\cite{HPST17}. (Theorem~\ref{thm:rmc})
	\item When  $\calF = \{A: w_1(A) \leq k_1, w_2(A) \leq k_2, \ldots, w_d(A) \leq k_d\}$ is defined by $d$ weight functions {\em and} each weight function $w_i$ is {\em poly-bounded}, then $\fpcm$ can be solved efficiently using dynamic programming.
	Therefore we get a $3$-approximation algorithm for the robust multi-knapsack center problem, extending the result in~\cite{HS86} to the case with outliers. (Theorem~\ref{thm:rmkn})
	\item When $\calF$ is given by the intersection of a single knapsack and a single matroid constraint, then we don't know the complexity. However, when the weight function $w(\cdot)$ is poly-bounded and the matroid is representable, then we can give a {\em randomized} algorithm for the $\fpcm$ problem via a reduction to the exact matroid intersection problem.
	Therefore, we get a randomized $3$-approximation for this special case of robust knapsack-and-matroid center problem (Theorem~\ref{thm:knandm}).
\end{asparaitem} \smallskip

\noindent
{\bf Remark 1: The Zero Outlier Case.}
At this juncture, the reader may wonder about the complexity of the $\calF$-center problem which doesn't allow any outliers. 
This is related to the following decision problem. Given $(X,\calF)$ and an arbitrary sub-partition $\cP$ of $X$, the problem asks whether there is a set $S\in \calF$ such that $S$ contains {\em exactly} one element from each part of $\cP$. We call this  the {\em $\calF$-feasibility under partition constraints} or simply the $\fpcf$ problem. Analogous to the informal theorem from earlier, the $\calF$-center problem (without outliers) has an efficient $3$-approximation algorithm if the $\fpcf$ problem can be solved efficiently; otherwise the $\calF$-center problem has no non-trivial approximation algorithm. Indeed, this theorem is much simpler to prove and arguably the roots of this lie in~\cite{HS86}. 

This raises the main open question from our paper: {\em what is the relation between the $\fpcf$ and the $\fpcm$ problem?} Clearly, the $\fpcf$ problem is as easy as the $\fpcm$ problem (set all values equal to one in the latter). But is there an $\calF$ such that $\fpcm$ is ``hard'' while $\fpcf$ is ``easy''? One concrete example is the corollary discussed in the last bullet point above. When $\calF$ is a single knapsack constraint and a single matroid constraint, then 
the $\fpcf$ problem is solvable in polynomial time by minimizing a linear function over a matroid polytope and another partition matroid {\em base} polytope. 
As noted above, we don't know the complexity of the $\fpcm$ problem in this case.

\noindent
{\bf Remark 2: Handling Approximations.}
If the $\fpcm$ problem is NP-hard, then the robust $\calF$-center has no non-trivial approximation algorithm. 
However approximation algorithms for $\fpcm$ translate to bi-criteria approximation algorithms for the robust $\calF$-center problem.
More precisely, if we have a $\rho$-approximation for the $\fpcm$ problem ($\rho \leq 1$), then we get a $(3,\rho)$-{\em bi-criteria} approximation algorithm for the robust $\calF$-center problem. That is, we return a solution $S\in \calF$ such that the maximum distance among the closest $\rho\cdot m$ points is at most $3$ times the optimum value.

There could be a different notion of approximation possible for the $\fpcm$ problem. Given an instance, there may be an algorithm which returns a set $S$ whose value is at least the optimum value but $S\in \calF^R$ for some $\calF^R\supseteq \calF$ which is a `relaxation' of $\calF$.
For instance, if $\calF$ is the intersection of multiple (constant) knapsack constraints which are not poly-bounded, then for any constant $\eps> 0$ the $\fpcm$ problem can be solved~\cite{CVZ11, GRSZ14} returning a set with value at least the optimum but violating each constraint by multiplicative $(1+\eps)$. We can use the same to get a polynomial time $3$-approximation for the robust multiple knapsack-center problem if we are allowed to violate the knapsack constraints by $(1+\eps)$.

\begin{table}[!ht]
	\setlength{\arrayrulewidth}{0.2mm}
	\setlength{\tabcolsep}{8pt}
	\renewcommand{\arraystretch}{1.5}
	\begin{center}
	\begin{tabular}{|c | c | c |} 
		\hline
		{\bf The constraint system $\calF$}& {\bf Without Outliers} & {\bf Robust (With Outliers)} \\  
		\hline
		Knapsack Constraint  & 3~\cite{HS86} & {\bf 3} (Theorem~\ref{thm:rkn}) \\ 
		\hline
		Matroid Constraint & 3~\cite{CLLW13} & 3~\cite{HPST17} \\
		\hline
		\makecell{Multiple  Knapsack \\ (poly-bounded weights)} & 3~\cite{HS86} & {\bf 3} (Theorem~\ref{thm:rmkn}) \\
		\hline
		Knapsack and Matroid & {\bf 3} (Theorem~\ref{thm:knandmcenter}) & \makecell{Open \\ {\bf 3} in special case} (Theorem~\ref{thm:knandm}) \\
		\hline
		\makecell{Multiple Knapsacks \\ and Matroid Constraint}&  \makecell{No uni-criteria \\ approximation}  & {\bf 3}, $(1+\varepsilon)$ violating (Theorem~\ref{thm:multiknapsackmat})\\ 
		\hline
	\end{tabular}
	\end{center}
\caption{All the above results can be obtained as corollary or simple extensions to our main result.
	The numbers in bold indicate new results.
}
\label{tbl:1}
\end{table}

\paragraph*{Our Technique}

Although our theorem statement is quite general, the proof is quite easy. Let us begin with the $\calF$-center problem without outliers.
For this, we follow the algorithmic `partitioning' idea outlined in paper~\cite{HS86} by Hochbaum and Shmoys. As is standard, we guess the optimum distance which we assume to be $1$ by scaling. Initially, all points are marked uncovered. Subsequently, we pick {\em any} uncovered point $x$ and consider a subset $B_x$ of points within distance $1$ from it. Note that the optimum solution {\em must} pick at least one point from each $B_x$ to serve $x$. Next, we call $x$ ``responsible'' for all uncovered points within a distance $2$ from it, and mark all these points covered. Observe that all the newly covered points are within distance $3$ from {\em any} point in $B_x$. 
We continue the above procedure till all points are marked covered. Also observe that the $B_x$'s form a sub-partition $\cP$ of the universe where each part has a responsible point. 
By the above two observations, we see that 
the $\fpcf$ problem must have a feasible solution with respect to $\cP$, and any solution to the $\fpcf$ problem gives a $3$-approximation to the $\calF$-center problem.

Handling outliers is a bit trickier. The above argument doesn't work since the `responsible' point may be an outlier in the optimal solution and we can no longer assert that the optimal solution must contain a point from each part. Indeed, the nub of the problem seems to be figuring out which points should be outliers. The $3$-approximation algorithm in~\cite{CKMN01} by Charikar \etal (see also paper~\cite{AF+10}) cleverly chooses the partitioning via a greedy procedure, but their argument seems hard to generalize to other constraints. 

A different attack used in the algorithm in~\cite{CGK16} by Chakrabarty \etal and that in~\cite{HPST17} by Harris \etal is by writing an LP relaxation and using the solution of the LP to recognize the outliers. At a high level, the LP assigns each point $x$ a variable (in this paper we call it $\cov(x)$) that indicates the extent to which $x$ is served. Subsequently, the partitioning procedure described in the first paragraph is run, except the responsible points are considered in decreasing order of $\cov(x)$. The hope is that points assigned higher $\cov(x)$ in the LP solution are less likely to be outliers, and therefore the partition returned by the procedure can be used to recover a $3$-approximate solution. 
This idea does work for the natural LP relaxation of the robust matroid center problem but fails for the natural LP relaxation of the robust knapsack center problem. Indeed, the latter has unbounded integrality gap. 

Our solution is to use the round-or-cut framework that has recently been a powerful tool in designing many approximation algorithms
(see ~\cite{CKK17,L16,L15,AMO14,CFLP00}). 
We consider the following ``coverage polytope'' for the  robust $\calF$-center problem:  the variables are $\cov(x)$ 
denoting the extent to which $x$ is covered by a convex combination of sets $S\in \calF$.
Of course, we cannot hope to efficiently check whether a particular $\cov$ lies in this polytope.
Nevertheless, we show that for any $\cov$ in the coverage polytope, the partitioning procedure when run in the decreasing order of $\cov$, has the property that there {\em exists} a solution $S\in \calF$ intersecting each part at most once which covers at least $m$ points. We can then use the algorithm for $\fpcm$ to find this set.  Furthermore, and more crucially, if the partitioning procedure does not have this property, then we can efficiently 
find a {\em hyperplane separating} $\cov$ from the the coverage polytope. Therefore, we can run the ellipsoid algorithm on the coverage polytope each time either obtaining a separating hyperplane, or obtaining a $\cov$ that leads to a desired partition, and therefore a $3$-approximation.

\section{Preliminaries}\label{sec:prelims}
In this section we give formal definitions and statements of our results.
As mentioned in a footnote in the Introduction, we focus on the supplier version of the problem. 

\begin{definition}[$\cF$-Supplier Problem]\label{deffc}
The input is a metric space $(X,d)$ on a set of points $X=F\cup C$ with distance function $d:X \times X \longrightarrow \mathbb{R}_{\geq 0}$ and $\cF \subseteq 2^F$ a down-closed family of subsets of $F$. The objective is to find $S \in \cF$ such that $\max_{v \in C} d(v,S)$ is minimized. 
\end{definition}

\begin{definition}[{\rfc} Problem]
	The input is an instance of the $\cF$-supplier problem along with an integer parameter $\sfm \in \{0,1,\ldots, |C|\}$.
	The objective is to find $S \in \cF$ and $T\subseteq C$ for which $\lvert T \rvert \geq \sfm$, and $\max_{u \in T} d(u,S)$ is minimized.
\end{definition}
\noindent
Thus an instance $\cI$ of the robust $\cF$-supplier problem is defined by the tuple $(F,C,d,\sfm,\calF)$.
In the definitions above, $F$ and $C$ are often called the set of \emph{facilities} and \emph{customers} respectively. 

Given the set system $\calF$ defined over $F$, we define the following optimization problem.

\begin{definition}[{$\fpcm$} problem]
The input is $\cJ = (F,\cF,\cP,\val)$ where $F$ is a finite set and $\cF \subseteq 2^F$ is a down-closed family, $\cP \subseteq 2^F$ is a sub-partition of $F$, and $\val:F \longrightarrow \{0,1,2,\cdots\}$ is an
integer-valued function with maximum range $|\val|$ 
satisfying: $\forall f_1,f_2 \in A \in \cP$, $\val(f_1) = \val(f_2)$. The objective is to find:
\[
\opt(\cJ)= \max\limits_{S \in \cF} ~~~ \val(S) : ~~~ \lvert S \cap A \rvert \leq 1, ~~\forall A \in \cP 
\]
\end{definition}
\noindent
The next theorem is the main result of the paper.
\begin{theorem}\label{mainthrm}\label{thm:1}
Given a {\rfc} instance $\cI = (F,C,d,\sfm,\cF)$, Let $\cA$ be an algorithm that solves any {$\fpcm$} instance $\cJ = (F,\cF,\cP,\val)$, with $|\val| \leq |C|$, in time bounded by $T_{\cA}(\cJ)$. Then, there is a $3$-approximation algorithm for the {\rfc} instance that runs in time $\mathrm{poly}(\lvert\cI\rvert)T_{\cA}(\cJ)$.
\end{theorem}

The next theorem is the (easier) second part of the dichotomy theorem. We show that if {$\fpcm$} cannot be solved, then the corresponding {\rfc} cannot be approximated.

\begin{theorem}\label{compthrm}\label{thm:2}
Given any non-trivial approximation algorithm $\cB$ for the {\rfc} problem that runs in time $\text{T}_{\cB}(|\cI|)$ on instance $\cI$, any {$\fpcm$} instance $\cJ = (F,\cF,\cP,\val)$ can be solved in time $\mathrm{poly}(\lvert\cJ\rvert)\text{T}_{\cB}(|\cI|)$, where $|\cI| = \poly(|\cJ|)$.
\end{theorem}

\begin{proof}
	Given $\cJ$ we construct an instance $\cI$ of the {\rfc} problem. The set of facilities is $F$.
	We describe the set of customers $C$ next. Extend $\cP$ to a partition of $F$ denoted by $\mathcal{Q} = \mathcal{P} \cup \{\{f\}: f\in F, \nexists A \in \mathcal{P} : f \in A\}$. By definition of the $\fpcm$ problem, for any $A \in \cQ$, there exists a number $n_A \in \{0,1,2,\cdots\}$ such that $\val(f) = n_A$, for all $f \in A$. 
For each $A \in \cQ$, we add $n_A$ customers to $C$ and call this set $\phi(A)$. 

We now describe the distance function. For each $A\in \cQ$, for each pair $u,v\in A$ and $u,v\in \phi(A)$ we have $d(u,v) = 0$.
For each $u\in A$ and $v\in \phi(A)$, we have $d(u,v) = 1$. All other distances are $\infty$. Observe that $d$ satisfies the triangle inequality.

Finally, we let $\sfm$ be our guess of the value of $\opt(\cJ)$. This completes the description of $\cI = (F,C,d,\sfm,\cF)$. 

Suppose algorithm $\cB$ finds $S \in \cF$ and $T \subseteq C$ such that $\lvert T \rvert \geq \sfm$ and $\max_{v \in T} d(v,S) \leq \alpha \opt(\cI) = \alpha$. Without loss of generality, we can assume $\lvert S \cap A \rvert \leq 1$ for all $A \in \cP$, which implies that $S$ is a feasible solution for $\cJ$. The reason is, if there exists $f_1,f_2\in S$ for which $f_1,f_2\in A \in \cP$, then $S\backslash f_2$ is still an $\alpha$-approximate solution for $\cI$. To see why this is true, recall that $\cF$ is down-closed so $S \backslash f_2 \in \cF$ and since $d(f_1,f_2) = 0$ then $S\backslash f_2$ covers all the customers that $S$ covers. 
Next, we assert that  $\val(S) \geq m = \opt(\cJ)$ since
$\sfm \leq \lvert T \rvert \leq \lvert \{v \in C : d(v,S) \leq \alpha\} \rvert = \sum_{A \in \cQ: \lvert S \cap A\rvert = 1} \lvert \phi(A)\rvert = \sum_{f \in S} \val(f),$
where the first equality uses the fact that for $v \in C$ and $f \in A \in \cQ$, $d(v,f) \leq \alpha$ only if $v \in \phi(A)$.

Finally, since $\val$ is poly-bounded which makes the value of $\opt(\cJ)$ to be bounded by $\mathrm{poly}(\lvert\cJ\rvert)$, one can iterate over all the possible values for $\opt(\cJ)$ to guess $\sfm$.
\end{proof}

We end this section by setting a few notations used in the remainder of the paper. 
For any $u\in F\cup C$ we let $B_C(u,r)$ be the customers in a ball of radius $r$ around $u$ i.e. $B_C(u,r) = \{v \in C : d(u,v) \leq r\}$. Similarly, define $B_F(u,r)$ as the facilities in a ball of radius $r$ around $u$ i.e. for $u\in F \cup C$, $B_F(u,r) = \{f \in F : d(u,f) \leq r\}$. 

\section{Algorithm and Analysis : Proof of Theorem~\ref{thm:1}}\label{mainApp}

We fix $\cI = (F,C,d,\cF, \sfm)$  the instance of the {\rfc} problem. 
We use $\hopt$ to denote \emph{our guess} of the value of the optimal solution. Without loss of generality, we can always assume $\hopt=1$ because if not, we could scale $d$ to meet this criteria.
Our objective henceforth is to either find a set $S \in \cF$ such that $\lvert \{ v \in C: d(v,S) \leq 1 \} \rvert \geq \sfm$, or prove
that $\opt(\cI) > 1$.

There are two parts to our proof. The first part is a partitioning procedure which given an assignment $\cov(v)\in \R_{\geq 0}$ for every customer 
$v\in C$, constructs an instance $\cJ$ of $\fpcm$. We call $\cov$ {\em valuable} if $\cJ$ has optimum value $\geq m$. Our procedure ensures that if $\cov$ is valuable, then we get a $3$-approximate solution for $\cI$.
This is described in Section~\ref{subsecRed}. 
The second part contains the proof of Theorem~\ref{thm:1}. In particular we show how using the round-and-cut methodology 
using polynomially many calls to $\cA$ (recall this is the algorithm for $\fpcm$) we can either prove $\opt(\cI) > 1$, or
find a valuable $\cov$. This is described in Section~\ref{sec:rndandcut}.

\subsection{Reduction to \texorpdfstring{$\fpcm$}{Lg}}\label{subsecRed}
Algorithm~\ref{alg:1} inputs an assignment $\{\cov(v) \in \R_{\geq 0} :v\in  C\}$.
It returns a sub-partition $\cP$
of $F$ and assigns $\val:F\to \{0,1,\cdots,|C|\}$ such that all the facilities in the same part of $\cP$ get the same $\val$. That is, 
it returns an $\fpcm$ instance $\cJ = (F,\cF, \cP, \val)$ with $|\val| \leq |C|$.

The algorithm maintains a set of \emph{uncovered} customers $U\subseteq C$ initialized to $C$ (Line~\ref{ln:1}). 
In each iteration, it picks the customer $v \in U$ with maximum $\cov$ (Line~\ref{ln:greedy}) and adds it to set $\reps$ (Line~\ref{ln:rep}). 
We add the set of facilities $B_F(v,1)$ at distance $1$ from $v$ to $\cP$ (Line~\ref{ln:nearby-fac},~\ref{ln:form-part}).
For each such $v$, we eke out the subset $\chld(v) = B_C(v,2) \cap U$ of currently uncovered clients ``represented'' by $v$ (Line~\ref{ln:chld}).
For every facility $f\in B_F(v,1)$ we define its \emph{value} to be: $\val(f) = \lvert \chld(v) \rvert$ (Line~\ref{ln:set-val}).
At the end of the iteration, $\chld(v)$ is removed from $U$ (Line~\ref{ln:12}) and the loop continues till $U$ becomes $\emptyset$.
This way, the algorithm partitions $C$ into $\{\chld(v) : v \in \reps\}$ (see fact\eqref{partfact}). 
Claim~\ref{subpartclm} shows that $\cP$ is a sub-partition of $F$. 

\begin{algorithm}[!ht]
		\caption{$\cF$-PCM instance construction}
		\label{constAlgo}\label{alg:1}
		\begin{algorithmic}[1]
			\Require {\rfc} instance $(F,C,d,\sfm,\cF)$ and assignment $\{\cov(v) \in \R_{\geq 0} :v\in  C\}$
			\Ensure $\cF$-PCM instance $(F,\cF,\mathcal{P},\val)$
			\State $U \leftarrow C$ \Comment{The set of uncovered customers} \label{ln:1}
			\State $\reps \leftarrow \emptyset$ \Comment{The set of representatives}
			\State $\mathcal{P} \leftarrow \emptyset$ \Comment{The sub-partition of $F$ that will be returned}
			\While{ $U \neq \emptyset$} 
	        \State $v \leftarrow \arg\max_{v\in U} \cov(v)$ \Comment{The first customer in $U$ in non-increasing $\cov$ order} \label{ln:6} \label{ln:greedy}
	        \State $\reps \leftarrow \reps \cup v$ \label{ln:7} \label{ln:rep}
	        \State $B_F(v,1) \leftarrow \{f \in F: d(f,v) \leq 1\}$ \Comment{Facilities that can cover $v$ with a ball of radius 1} \label{ln:8} \label{ln:nearby-fac}
			\State $\mathcal{P} \leftarrow \mathcal{P} \cup B_F(v,1)$ \label{ln:9} \label{ln:form-part}
	        \State $\chld(v) \leftarrow \{u \in U: d(u,v) \leq 2\}$\Comment{Equals to $B_C(v,2)\cap U$} \label{ln:chld}	        \label{ln:10}
	        \State $\val(f) \leftarrow \lvert \chld(v) \rvert \ \ \forall f \in B_F(v,1)$  \label{ln:11} \label{ln:set-val}
	        \State $U \leftarrow U \backslash \chld(v)$ \label{ln:12} \label{ln:remove-from-U}
			\EndWhile
		\end{algorithmic}
	\end{algorithm}
	
	\begin{fact}\label{partfact}
	$\{\chld(v) : v \in \reps\}$ is a partition of $C$.
	\end{fact}
	
	\begin{fact}\label{greedyrule}
	For a $v \in \reps$ and any $u \in \chld(v)$ line 6 of the algorithm implies $\cov(v) \geq \cov(u)$.
	\end{fact}
	
	\begin{claim}\label{subpartclm}
	$\mathcal{P}$ constructed by Algorithm~\ref{constAlgo} is a sub-partition of $F$.
	\end{claim}
	\begin{proof}
	By Line~\ref{ln:remove-from-U} of the algorithm, for each $u,v \in \reps$ we have $d(u,v) > 2$ hence $B_F(u,1) \cap B_F(v,1) = \emptyset$ implying $\mathcal{P}$ is a sub-partition of $F$.
	\end{proof}
	
	\begin{claim}\label{bll3clm}
	For each $v \in \reps$ and $f \in B_F(v,1)$, $\chld(v) \subseteq B_C(f,3)$.
	\end{claim}
	\begin{proof}
	For any $u\in \chld(v)$, we have $d(u,v) \leq 2$ and since $d(f,v) \leq 1$, the fact that $d$ is metric implies $d(f,u) \leq 3$.
	\end{proof}
	
	\begin{definition}\label{ISdef}
	For $S \subseteq F$ let $R(S) = \{v \in \reps : B_F(v,1) \cap S \neq \emptyset \}$, be the set of representative customers in $\reps$ that are covered by balls of radius 1 around the facilities in $S$.
	\end{definition}
	
	\begin{claim}\label{valclm}
	Let $S \in \cF$ be any feasible solution of the $\cF$-PCM instance constructed by Algorithm~\ref{constAlgo}. Then,
	    $\sum_{f \in S} \val(f) = \sum_{v\in R(S)} \lvert \chld(v) \rvert.$
	\end{claim}
	\begin{proof}
	For an $f \in S$, according to Line~\ref{ln:set-val} of the algorithm, $\val(f) > 0$ only if $f \in B_F(v,1)$ for some $v \in \reps$. Also, by definition of the $\cF$-PCM problem, $\lvert B_F(v,1) \cap S\rvert\ \leq 1$ for any $v \in \reps$. That is, there is exactly one $f \in B_F(v,1) \cap S$ for each $v \in R(S)$ and again by line~\ref{ln:set-val}, $\val(f) = \lvert \chld(v) \rvert$. Summing this equality over all $v \in R(S)$ and the corresponding $f \in B_F(v,1) \cap S$ proves the claim.
	\end{proof}

		\begin{claim}\label{slnConstLma}
			Let $\cI = (F,C,d,\sfm,\cF)$ be a {\rfc} instance and let $\cov:C\to \mathbb{R}_{\geq 0}$ be a coverage function.
			Let $\cJ = (F,\cF,\cP,\val)$ be the $\fpcm$ instance returned by Algorithm~\ref{constAlgo} on input $\cI$ and $\cov$.
			Given any feasible solution $S$ to $\cJ$, we can cover at least $\val(S)$ customers of $C$ by opening radius $3$-balls around each facility in $S$.

	\end{claim}
	
	\begin{proof}
		By considering $R(S)$ from Definition~\ref{ISdef}, Claim~\ref{valclm} gives: 
		$\sum_{v\in R(S)} \lvert \chld(v) \rvert=\sum_{f \in S} \val(f)$. 
		From Fact~\ref{partfact}, we get that for all $u,v \in \reps, \chld(u) \cap \chld(v) = \emptyset$.
		Thus, $\lvert \bigcup_{v\in R(S)} \chld(v)\rvert = \sum_{v \in R(S)} \lvert \chld(v) \rvert = \val(S)$.
		Furthermore, by
		Claim~\ref{bll3clm}, $\{ v \in C: d(v,S) \leq 3 \} \supseteq \bigcup_{u\in R(S)} \chld(u)$ 
		implying the size of the former is at least $\val(S)$, thus proving the lemma.

	\end{proof}
		The above claim motivates the following definition of {\em valuable} $\cov$ assignments, and the subsequent lemma.
		\begin{definition} 
		An assignment $\{\cov(v) \in \R_{\geq 0} :v\in  C\}$ 
		is \emph{valuable} with respect to a {\rfc} instance $\cI = (F,C,d,\sfm,\cF$), iff $\opt(\cJ) \geq \sfm$, where $\cJ$ is the $\cF$-PCM instance returned by Algorithm~\ref{constAlgo} from $\cI$ and $\cov$.
	\end{definition}
\begin{lemma}\label{lem:summary}
	Given an instance $\cI$ of the {\rfc} problem with $\opt(\cI) = 1$, and a valuable assignment $\cov$ with respect to it, we can obtain a $3$-approximate solution in time $\poly(|\cI|) + T_\cA(\cJ)$ where $\cJ$ is the instance constructed by Algorithm~\ref{constAlgo} from $\cI$ and $\cov$.
\end{lemma}
\begin{proof}
Since $\cov$ is valuable, $\opt(\cJ)\geq m$.
We use solver $\cA$ to return an optimal solution $S\in \cF$ with $\val(S) \geq m$.
Claim~\ref{slnConstLma} implies that $S$ is a $3$-approximate solution to $\cI$.
\end{proof}
	
\subsection{The Round and Cut Approach}\label{subsecRound}\label{sec:rndandcut}
If the guess $\hopt = 1$ for $\cI = (F,C,d,\sfm,\cF)$ is at least $\opt(\cI)$, then the following polytope must be non-empty.
To see this, if $S^* \in \cF$ is the optimal solution to $\cI$ then set $z_{S^*} := 1$ and $z_S := 0$ for $S \in \cF \backslash S^*$.

\begin{alignat}{4}
\CovP = \{(\cov(v): v\in C) : 
&& \sum\limits_{v\in C} \cov(v) & \geq & ~~\sfm \tag{$\CovP$.1} \label{eq:P1} \\
\forall v\in C, && ~~\cov(v) - \sum\limits_{\substack{S \in \cF: d(v,S) \leq 1}} z_S &=& ~~0 \tag{$\CovP$.2} \label{eq:P2} \\
&& 	\sum\limits_{S \in \cF} z_S & = & ~~ 1  \tag{$\CovP$.3} \label{eq:P3} \\
\forall S\in \cF,  && z_S &\geq &0\} \notag \tag{$\CovP$.4} \label{eq:P4}
\end{alignat}
Even though $\CovP$ has exponentially many auxiliary variables ($z_S$ for all $S\in \cF$), its dimension is still $|C|$. The following gives a family of valid inequalities for $\CovP$ via Farkas lemma.
\begin{lemma}\label{prehypLma}
Let $\lambda(v)\in \mathbb{R}$ for every $v\in C$ be such that 
\begin{equation}\label{prop1}
    \sum\limits_{ \substack{v \in C:\\ d(v,S) \leq 1}} \lambda(v) \leq \sfm \ \ \ \forall S \in \cF \tag{V1}
\end{equation}
Then any $\cov\in \CovP$ satisfies
\begin{equation}\label{prop2}
    \sum\limits_{v \in C} \lambda(v)\cov(v) \leq \sfm \tag{V2}
\end{equation}
\end{lemma}
\begin{proof}
Given $\cov \in \CovP$, there exists $\{z_S: S\in \cF\}$ such that together they satisfy \eqref{eq:P1}-\eqref{eq:P4}.
\begin{align*}
\sum\limits_{v \in C} \lambda(v)\cov(v) &~~=_{\eqref{eq:P2}} ~~~ \sum\limits_{v \in C} \lambda(v)\sum\limits_{ \substack{S \in \cF:\\ d(v,S) \leq 1}} z_S  ~~= \sum\limits_{S \in \cF} z_S\sum\limits_{ \substack{v \in C:\\ d(v,S) \leq 1}} \lambda(v)\\
&~~\leq_{\eqref{prop1},\eqref{eq:P4}}  ~~~ \sfm \sum\limits_{S \in \cF} z_S =_{\eqref{eq:P3}} ~~~\sfm  
\end{align*}
\end{proof}

\noindent
The next lemma shows that all $\cov$'s in $\CovP$ are valuable.
\begin{lemma}
\label{hypLma}
Suppose an assignment $\{\cov(v) \in \R_{\geq 0} :v\in  C\}$ is not valuable with respect to $\cI= (F,C,d,\sfm,\cF)$.
Then there is a hyper-plane separating it from $\CovP$ that can be constructed in polynomial time.
\end{lemma}
\begin{proof}
	If $\sum_{v\in C} \cov(v) < \sfm$, this inequality itself is a separating hyper-plane and we are done. So we may assume $\sum_{v\in C} \cov(v)\geq \sfm$.

	Let $\cJ = (F,\cF,\cP,\val)$ be the $\cF$-PCM instance constructed by Algorithm~\ref{constAlgo} from $\cI$ and $\cov$. Fix $S \in \cF$ and recall from Definition~\ref{ISdef} that $R(S) = \{v \in \reps : B_F(v,1) \cap S \neq \emptyset \}$. Pick an arbitrary $T\subseteq S$ for which $\lvert B_F(v,1) \cap T\rvert = 1$, for all $v \in R(S)$. Observe that by down-closedness of $\cF$, we have $T\in \cF$ which implies $T$ is a feasible solution for $\cJ$, and since $\cov$ is not valuable $\val(T) < \sfm$.
	Furthermore, Claim~\ref{valclm} applied to $T$ gives $\val(T) = \sum_{v \in R(T)} |\chld(v)|$.
	Since $R(S)=R(T)$ and $|\chld(v)|$ is integer-valued, we get:
	\begin{equation}\label{lowp}
	\sum_{v\in R(S)} \lvert \chld(v) \rvert \leq \sfm -1
	\end{equation}
	Let $\alpha = \frac{\sfm}{\sfm - 0.5} > 1$. Define $\lambda(v)$ for $v \in C$ as:
\[\lambda(v) = \begin{cases}
	\alpha \lvert\chld(v)\rvert & v \in \reps\\
	0 &  \textrm{for all other} ~v\in C
	\end{cases} \]
Now observe that for any $S\in \cF$:
	\begin{equation*}
	\sum\limits_{ \substack{v \in C: d(v,S) \leq 1}} \lambda(v) = \sum\limits_{ \substack{v \in \reps:d(v,S) \leq 1}}  \alpha \lvert\chld(v)\rvert = \alpha \sum\limits_{v \in R(S)} \lvert\chld(v)\rvert \leq \alpha(\sfm-1) < \sfm 
	\end{equation*}
That is, $\lambda(v)$'s satisfy~\eqref{prop1}.
	Now we prove \eqref{prop2} is not satisfied thus it can be used to separate $\cov$ from $\CovP$.
\begin{align}
\sum\limits_{v \in C} \lambda(v)\cov(v) & = & \alpha\sum\limits_{v \in \reps} \lvert\chld(v)\rvert\cov(v) & ~~= & \alpha\sum\limits_{v \in \reps}\sum\limits_{u \in \chld(v)}\cov(v) \notag \\
& \geq_{\textrm{Fact}~\ref{greedyrule}} & \alpha\sum\limits_{v \in \reps}\sum\limits_{u \in \chld(v)}\cov(u) & ~~=_{\textrm{Fact}~\ref{partfact}} & \alpha\sum\limits_{v \in C}\cov(v) \geq \alpha m > m \notag
\end{align}	
\end{proof}

\begin{proof}[{\bf Proof of Theorem~\ref{mainthrm}}]
Given the guess $\hopt$ which is scaled to $1$, we use the ellipsoid algorithm to check if $\CovP$ is empty or not.
Whenever ellipsoid asks if a given $\cov$ is in $\CovP$ or not, run Algorithm~\ref{constAlgo} for this given $\cov$ to construct the corresponding $\cF$-PCM instance $\cJ$ and use algorithm $\cA$, promised in the statement of Theorem~\ref{mainthrm}, to solve it. 
If $\opt(\cJ) \geq \sfm$, then Lemma~\ref{lem:summary} implies that we have a $3$-approximate solution. 
Otherwise, $\cov$ is not valuable, and we can use Lemma~\ref{hypLma} to find a separating hyperplane. 
In polynomial time, either we get a $\cov\in \CovP$ which by Lemma~\ref{hypLma} has to be valuable, or we prove $\CovP$ is empty and we modify our $\hopt$ guess. For the correct guess, the latter case won't occur and we get a $3$-approximate solution.
\end{proof}

\section{Applications and Extensions}\label{secApp}

In this section we elaborate on the applications and extensions stated in the Introduction.
We begin with looking at specific instances of $\cF$ which have been studied in the literature, and some which have not. \medskip

\noindent
{\bf Single and Multiple Knapsack Constraints.} We look at 
\[\cF_{\mathsf{KN}} := \{S\subseteq F: \text{ for $i=1,\ldots, d$}, ~\sum_{v \in S} w_i(v) \leq k_i\}\] 
where there are $d$ weight functions over $F$ and $k_i$'s are upper bounds on these weights.
Of special interest is the case $d=1$ in which we get the robust knapsack supplier problem also called the weighted $k$-supplier problem with outliers.

The $\fpcm$ problem for the above $\cF_{\mathsf{KN}}$ has the following complexity: When $d=1$, the problem can be solved in polynomial time. Indeed, given a partition $\cP$, since $\val(u) = \val(v)$ for all $v$ in the same part, any solution which picks a facility from a part $A \in \cP$ may as well pick the one with the smallest weight in that part. Thus, the problem boils down to the usual knapsack problem in which we have $|\cP|$ items where the item corresponding to part $A\in \cP$ has weight $\min_{v\in A} w(v)$ and value $\val(v)$.
Since the values are poly-bounded, this problem is solvable in polynomial time. Thus, we get the following corollary to Theorem~\ref{thm:1} resolving the open question raised in~\cite{CLLW13} and~\cite{HPST17}.
\begin{theorem}\label{thm:rkn}
	There is a polynomial time $3$-approximation to the robust knapsack center problem.
\end{theorem}
When $d > 1$, then the $\fpcm$ problem is NP-hard even when $\val$ is poly-bounded. However, if the $w_i$'s are also poly-bounded (actually one of them can be general), then the $\fpcm$ problem can be solved in polynomial time using dynamic programming. This problem was in fact studied in ~\cite{HS86} (the conference version) and is called the {\em suitcase} problem there.
Thus, we get the following corollary to Theorem~\ref{thm:1} extending the result in~\cite{HS86}.
\begin{theorem}\label{thm:rmkn}
	There is a polynomial time $3$-approximation to the robust multiple-knapsack center problem if the number of weights is a constant and all but possibly one weight function are poly-bounded.
\end{theorem}

\noindent
{\bf Single and Multiple Matroid Constraints.} 
We look at 
\[\cF_{\mathsf{Mat}} := \{S\subseteq F: S\in \cI_{M_i}, ~\forall i=1,\ldots,d\}\] 
When $d=1$, we get the robust matroid center problem. The $\fpcm$ paper reduces to finding a maximum value set in $\cI_M$ and a partition matroid induced by $\cP$. This is solvable in polynomial time even when $\val$ is general and not poly-bounded, and even when $\cI_M$ is given as an independent set oracle.
Thus, we get the following corollary to 
Theorem~\ref{thm:1} obtaining the result in~\cite{HPST17}.
\begin{theorem}\label{thm:rmc}[Theorem 1.1 in~\cite{HPST17}]
	There is a polynomial time $3$-approximation to the robust matroid center problem even when the matroid is described as an independent set oracle.
\end{theorem}
When there are $d>1$ matroids, then the $\fpcm$ problem is NP-hard. Therefore, Theorem~\ref{thm:2} implies that for instance, we can have {\em no} unicriteria approximation for the robust matroid-intersection center problem. \medskip

\noindent
{\bf Single Knapsack and Single Matroid Constraint.}
We look at 
\[\cF_{\mathsf{KN}\cap\mathsf{Mat}} := \{S\subseteq F: \sum_{v\in S} w(v) \leq k, ~~ S\in \cI_{M}\}\] 
which is the intersection of a single matroid and a single knapsack constraint. To the best of our knowledge, the resulting {\rfc} problem has not been studied before. One natural instantiation is when $F$ is a collection of high-dimensional vectors with weights and the constraint on the centers is to pick a linearly independent set with total weight at most $k$.

The corresponding $\fpcm$ problem asks us, given a partition $\cP$ and poly-bounded values $\val$, to find a set $S\in \cI_\cM \cap \cI_{\cP}$ of maximum value such that $w(S)\leq k$, where $\cI_{\cP}$ is the partition matroid induced by $\cP$. We don't know if this problem can be solved in polynomial time, even in the case when $M$ is another partition matroid.

\def\bw{\mathbf{w}}
\def\bW{\mathbf{W}}

However, the above problem is related to the {\em exact matroid intersection} problem. In this problem, we are given two matroids $\cM$ and $\cP$, and a weight function $\bw$ on each ground element and a budget $\bW$. The objective is to decide whether or not there is a set $S \in \cI_\cM \cap \cI_P$ such that $\bw(S) = \bW$. Understanding the complexity of this problem is a long standing challenge~\cite{Cam92,MVV87,PY82}. When the matroids are representable over the same field, then ~\cite{Cam92} gives a randomized pseudopolynomial time algorithm for the problem. The following claim shows the relation between $\fpcm$ and the exact matroid intersection problem; this claim is essentially present in ~\cite{BerBFS08}.

\begin{claim}\label{clm:ipco}
Given an algorithm for the exact matroid intersection problem, one can solve the $\fpcm$ problem in polynomial time when the weights $w$ are poly-bounded.
\end{claim}

\begin{proof}
We guess $\sfv^*$ to be the optimum value of the $\fpcm$ problem; since $\val$ is poly-bounded, there are only polynomially many guesses.
We also guess $k^* \leq k$ to be the total $w$ of the optimum set. Again if $w$ is poly-bounded, there are polynomially many guesses.
We define a weight function $\bw$ as follows. Let $\phi = w(F) + 1$ be a large enough upper-bound on the possible values of $w(S), S\subseteq F$. Define $\bw(f) = \phi \val(f) + w(f)$ for all $f\in F$ and $\bW = \phi \sfv^* + k^*$.

 We claim that there is a set $S$ in $\cI_M \cap \cI_{\cP}$ with $\bw(S) = \bW$ iff $\val(S) = \sfv^*$ and $w(S) = k^*$.
 The if-direction is trivial. 

 On the other hand if $\bw(S) = \bW$ we get
 $k^* = \bw(S) - \phi \sfv^* = \phi \val(S) + w(S) - \phi V^*$. Now if $\val(S) \neq \sfv^*$ since $\val$ is integer-valued and since $\phi > w(S)$ for any $S\subseteq F$, 
 the RHS is either negative or $> w(F)$. In any case it cannot be $k^*$. Therefore, we must have $\val(S) = \sfv^*$ which implies $w(S) = k^*$.
\end{proof}

Armed with the non-trivial result about exact matroid intersection from~\cite{Cam92}, we get the following.

\begin{theorem}\label{thm:knandm}
	Given a linear matroid $\cM$ and a poly-bounded weight function, there is a 
randomized polynomial time $3$-approximation to the robust knapsack-and-matroid center problem.
\end{theorem}

\subsection{The Case of No Outliers}\label{subsecapx}
The $\cF$-supplier problem, that is the case of $\sfm = |C|$, may be of special interest. In this case the problem is easier and the complexity is defined by the complexity of the following decision problem.
\begin{definition}[{$\fpcf$} problem]
	The input is $\cJ = (F,\cF,\cP)$ where $F$ is a finite set, $\cF \subseteq 2^F$ is a down-closed family and  $\cP \subseteq 2^F$ is an arbitrary sub-partition of $F$. The objective is to decide whether there {\em exists} a set $S\in \cF$ such that
	$\lvert S \cap A \rvert = 1, ~~\forall A \in \cP $.
\end{definition}
\begin{theorem}\label{thm:1b}
If the $\fpcf$ problem can be solved efficiently for any partition $\cP$, then 
the $\cF$-supplier problem has a polynomial time $3$-approximation. Otherwise, there is no non-trivial approximation possible
for the $\cF$-supplier problem.
\end{theorem}
\begin{proof}[Sketch]
	Run Algorithm~\ref{constAlgo} with an arbitrary assignment $\cov$ (and ignore the $\val$'s). Let $\cJ = (F,\cF,\cP)$ be the resulting
	$\fpcf$ instance. If the guess $\hopt=1$ is correct, then note that the optimum solution $S^*$ must satisfy $S^*\cap A \neq \emptyset$
	for all $A\in \cP$; if not, then the corresponding $v\in \reps$ can't be served. Conversely, any $S$ satisfying $S\cap A\neq\emptyset$ for all $A\in \cP$ implies a $3$-approximate solution.
	Therefore, an algorithm for $\fpcf$ can either give a $3$-approximate solution or prove the guess $\hopt$ is too low.
\end{proof}

Theorem~\ref{thm:1} and Theorem~\ref{thm:1b} raise the question: is there any set of constraints for which the problem without outliers is significantly easier than the problem with outliers? We don't know the answer to this question, although we guess the answer is yes.
For this, it suffices to design a set system for which $\fpcf$ is easy but $\fpcm$ is hard (perhaps NP-hard). To see the difference between these problems consider the $\cF_{\mathsf{KN}\cap\mathsf{Mat}}$ family described in the previous subsection. We don't know if $\fpcm$ is easy or hard, but $\fpcf$ is easy: this amounts to minimizing $w(S)$ over $S\in \cI_\cM \cap \cB_\cP$ where $\cB_\cP$ is the base polytope induced by $\cP$. This can be done in polynomial time, and therefore we get the following corollary.

\begin{theorem}\label{thm:knandmcenter}
	There is a polynomial time $3$-approximation to the knapsack-and-matroid center problem.
\end{theorem}

\subsection{Handling Approximation}
The technique used to prove Theorem~\ref{thm:1} is robust enough to translate approximation algorithms for the $\fpcm$ problem to
{\em bi-criteria} approximation algorithms for the {\rfc} problem. There are two notions of approximation algorithms for the $\fpcm$ problem
and they lead to two notions of bi-criteria approximation.

The first is the standard notion: a $\rho$-approximation (for $\rho \le 1$) algorithm that takes instance $\cJ$ of $\fpcm$, returns a solution 
$S\in \cF$ of value $\val(S) \ge \rho \cdot \opt(\cJ)$. The corresponding bi-criteria approximation notion for the {\rfc} problem is the following: an $(\alpha,\beta)$-approximation algorithm for instance $\cI$ of {\rfc} returns a solution which opens centers at $S\in \cF$
and the distance of at least $\beta \sfm$ customers to $S$ is $\leq \alpha \cdot \opt(\cI)$. The proof of Theorem~\ref{thm:1} in fact implies the following.
\begin{theorem}
	\label{thm:1c}
	Let $\cA$ be a polynomial time $\rho$-approximate algorithm for the $\fpcm$ problem. Then there is a
	 polynomial time $(3,\rho)$-bi-criteria approximation algorithm for the {\rfc} problem. 
\end{theorem}

The second notion of approximation for the $\fpcm$ problem is one which satisfies the constraints approximately. This notion is more problem dependent and makes sense only if there is a notion of an approximate relaxation $\cF^R$ for the set $\cF$. For example, an $(1+\eps)$-relaxation for $\cF_{\mathsf{KN}}$ could be the subsets $S$ with 
$w_i(S) \leq (1+\eps)\cdot k_i$ for all $i$. A $\rho$-violating algorithm for an instance  $\cJ$ of $\fpcm$ would then return a set $S$ with $\val(S)\geq \opt(\cJ)$ but $S\in \cF^R$ which is an $\rho$-relaxation for $\cF$. 
This defines a different bi-criteria approximation notion for the {\rfc} problem. An $\alpha$-approximate $\beta$-violating algorithm for the {\rfc} problem takes an instance $\cI$ and returns a solution $S\in \cF^R$ which is a $\beta$-relaxation for $\cF$ such that
at least $m$ customers in $C$ are at distance at most $\alpha\cdot \opt(\cI)$ to $S$.

\begin{theorem}
	\label{thm:1d}
	Let $\cA$ be a polynomial time $\rho$-violating algorithm for the $\fpcm$ problem. Then there is a
	polynomial time $3$-approximate-$\rho$-violating algorithm for the {\rfc} problem.
\end{theorem}

When $\calF$ is described by constant $d$ knapsack constraints (with general weights) and  a single matroid constraint, for any constant $\eps>0$
Chekuri~\etal give an $(1+\eps)$-approximation algorithm for the $\fpcm$ in~\cite{CVZ11}.
Without the matroid constraint, Grandoni~\etal give an $(1+\eps)$-violating algorithm in~\cite{GRSZ14}.
Together, we get the following corollary. The latter recovers a result from~\cite{CLLW13}.
\begin{theorem}\label{thm:multiknapsackmat}
	Fix any constant $\eps > 0$. There is a polynomial time $(3,(1+\eps))$-bi-criteria approximation algorithm for the robust supplier problem with 
	constant many knapsack constraints and one matroid constraint. There is a polynomial time $3$-approximate $(1+\eps)$-violating algorithm for the robust supplier problem with constant many knapsack constraints.
\end{theorem}
\bibliographystyle{alpha}
\bibliography{lipics-v2018-sample-article}

\newcommand{\etalchar}[1]{$^{#1}$}
\begin{thebibliography}{CKMN01}

\bibitem[APF{\etalchar{+}}10]{AF+10}
Gagan Aggarwal, Rina Panigrahy, Tom{\'{a}}s Feder, Dilys Thomas, Krishnaram
  Kenthapadi, Samir Khuller, and An~Zhu.
\newblock {Achieving Anonymity via Clustering}.
\newblock {\em {ACM} Trans. Algorithms}, 6(3):49:1--49:19, 2010.

\bibitem[ASS14]{AMO14}
Hyung-Chan An, Mohit Singh, and Ola Svensson.
\newblock {LP}-based algorithms for capacitated facility location.
\newblock In {\em Proceedings of the 2014 IEEE 55th Annual Symposium on
  Foundations of Computer Science}, FOCS '14, pages 256--265, Washington, DC,
  USA, 2014. IEEE Computer Society.

\bibitem[BBGS11]{BerBFS08}
Andr{\'e} Berger, Vincenzo Bonifaci, Fabrizio Grandoni, and Guido Sch{\"a}fer.
\newblock {Budgeted Matching and Budgeted Matroid Intersection via the Gasoline
  Puzzle}.
\newblock {\em Mathematical Programming}, 128(1):355--372, 2011.

\bibitem[CFLP00]{CFLP00}
Robert~D. Carr, Lisa~K. Fleischer, Vitus~J. Leung, and Cynthia~A. Phillips.
\newblock {Strengthening {I}ntegrality Gaps for Capacitated Network Design and
  Covering Problems}.
\newblock In {\em Proceedings of the Eleventh Annual ACM-SIAM Symposium on
  Discrete Algorithms}, SODA '00, pages 106--115, Philadelphia, PA, USA, 2000.
  Society for Industrial and Applied Mathematics.

\bibitem[CGK16]{CGK16}
Deeparnab Chakrabarty, Prachi Goyal, and Ravishankar Krishnaswamy.
\newblock {The Non-Uniform $k$-Center Problem}.
\newblock In {\em 43rd International Colloquium on Automata, Languages, and
  Programming, {ICALP} 2016, July 11-15, 2016, Rome, Italy}, pages 67:1--67:15,
  2016.

\bibitem[CGM92]{Cam92}
Paolo~M. Camerini, Giulia Galbiati, and Francesco Maffioli.
\newblock {Random Pseudo-Polynomial Algorithms for Exact Matroid Problems}.
\newblock {\em J. Algorithms}, 13(2):258--273, 1992.

\bibitem[CKK17]{CKK17}
Deeparnab Chakrabarty, Ravishankar Krishnaswamy, and Amit Kumar.
\newblock {The Heterogeneous Capacitated k-Center Problem}.
\newblock In {\em Integer Programming and Combinatorial Optimization - 19th
  International Conference, {IPCO} 2017, Waterloo, ON, Canada, June 26-28,
  2017, Proceedings}, pages 123--135, 2017.

\bibitem[CKMN01]{CKMN01}
Moses Charikar, Samir Khuller, David~M. Mount, and Giri Narasimhan.
\newblock {Algorithms for Facility Location Problems with Outliers}.
\newblock In {\em Proceedings of the Twelfth Annual ACM-SIAM Symposium on
  Discrete Algorithms}, SODA '01, pages 642--651, Philadelphia, PA, USA, 2001.
  Society for Industrial and Applied Mathematics.

\bibitem[CLLW13]{CLLW13}
Danny~Z. Chen, Jian Li, Hongyu Liang, and Haitao Wang.
\newblock {Matroid and Knapsack Center Problems}.
\newblock In {\em Integer Programming and Combinatorial Optimization (IPCO)},
  pages 110--122, 2013.

\bibitem[CVZ11]{CVZ11}
Chandra Chekuri, Jan Vondrák, and Rico Zenklusen.
\newblock {Multi-budgeted Matchings and Matroid Intersection via Dependent
  Rounding}.
\newblock {\em Proceedings of the Twenty-Second Annual ACM-SIAM Symposium on
  Discrete Algorithms}, pages 1080--1097, 2011.

\bibitem[Gon85]{Gon85}
Teofilo~F. Gonzalez.
\newblock {Clustering to Minimize the Maximum Intercluster Distance}.
\newblock {\em Theoretical Computer Science}, 38:293 -- 306, 1985.

\bibitem[GRSZ14]{GRSZ14}
Fabrizio Grandoni, R.~Ravi, Mohit Singh, and Rico Zenklusen.
\newblock {New Approaches to Multi-Objective Optimization}.
\newblock {\em Math. Program.}, 146(1-2):525--554, 2014.

\bibitem[Hak64]{Hakimi64}
Seifollah~L. Hakimi.
\newblock {Optimum Locations of Switching Centers and the Absolute Centers and
  Medians of a Graph}.
\newblock {\em Ops Res.{INFORMS}}, 12:450--459, 1964.

\bibitem[Hak65]{Hakimi65}
Seifollah~L. Hakimi.
\newblock {Optimum Distribution of Switching Centers in a Communication Network
  and Some Related Graph Theoretic Problems}.
\newblock {\em Ops Res.{INFORMS}}, 13:462--475, 1965.

\bibitem[HN79]{HN79}
Wen-Lian Hsu and George~L. Nemhauser.
\newblock {Easy and Hard Bottleneck Location Problems}.
\newblock {\em Discrete Applied Mathematics}, 1(3):209 -- 215, 1979.

\bibitem[HPST17]{HPST17}
David~G. Harris, Thomas Pensyl, Aravind Srinivasan, and Khoa Trinh.
\newblock {A Lottery Model for Center-Type Problems with Outliers}.
\newblock In {\em Approximation, Randomization, and Combinatorial Optimization.
  Algorithms and Techniques, {APPROX/RANDOM} 2017, August 16-18, 2017,
  Berkeley, CA, {USA}}, pages 10:1--10:19, 2017.

\bibitem[HS85]{HS85a}
Dorit~S. Hochbaum and David~B. Shmoys.
\newblock {A Best Possible Heuristic for the $k$-Center Problem}.
\newblock {\em Math. Oper. Res.}, 10(2):180--184, 1985.

\bibitem[HS86]{HS86}
Dorit~S. Hochbaum and David~B. Shmoys.
\newblock {A Unified Approach to Approximation Algorithms for Bottleneck
  Problems}.
\newblock {\em J. ACM}, 33(3):533--550, 1986.

\bibitem[Li15]{L15}
Shi Li.
\newblock {On Uniform Capacitated $k$-median Beyond the Natural {{LP}}
  Relaxation}.
\newblock In {\em Proceedings of the Twenty-sixth Annual ACM-SIAM Symposium on
  Discrete Algorithms}, SODA '15, pages 696--707, Philadelphia, PA, USA, 2015.
  Society for Industrial and Applied Mathematics.

\bibitem[Li16]{L16}
Shi Li.
\newblock {Approximating Capacitated $k$-median with $(1 + \eps)k$ Open
  Facilities}.
\newblock In {\em Proceedings of the Twenty-seventh Annual ACM-SIAM Symposium
  on Discrete Algorithms}, SODA '16, pages 786--796, Philadelphia, PA, USA,
  2016. Society for Industrial and Applied Mathematics.

\bibitem[MVV87]{MVV87}
Ketan Mulmuley, Umesh~V. Vazirani, and Vijay~V. Vazirani.
\newblock {Matching is As Easy As Matrix Inversion}.
\newblock In {\em Proceedings of the Nineteenth Annual ACM Symposium on Theory
  of Computing}, STOC '87, pages 345--354, New York, NY, USA, 1987. ACM.

\bibitem[PY82]{PY82}
Christos~H. Papadimitriou and Mihalis Yannakakis.
\newblock {The Complexity of Restricted Spanning Tree Problems}.
\newblock {\em J. ACM}, 29(2):285--309, April 1982.

\end{thebibliography}
\end{document}